\documentclass[10pt,conference]{IEEEtran}
\usepackage{latexsym, graphicx, pstricks, pst-node, url, subfig, graphics, epic, bm}
\usepackage{amssymb}
\usepackage{amsthm}
\usepackage{amsmath}
\usepackage[left=0.75in,right=0.75in,top=0.75in,bottom=1in]{geometry}
\newtheorem{definition}{Definition}

\newtheorem{lemma}{Lemma}
\newtheorem{theorem}{Theorem}
\newtheorem{corollary}{Corollary}

\begin{document}

\title{On the minimum weight problem of permutation codes under Chebyshev distance}

\author{
\authorblockN{Min-Zheng Shieh}
\authorblockA{Department of Computer Science\\
National Chiao Tung University \\
1001 University Road, Hsinchu, Taiwan\\
Email: mzhsieh@csie.nctu.edu.tw}
\and
\authorblockN{Shi-Chun Tsai}
\authorblockA{Department of Computer Science\\
National Chiao Tung University \\
1001 University Road, Hsinchu, Taiwan\\
Email: sctsai@cs.nctu.edu.tw}
}
%

\maketitle

\begin{abstract}
Permutation codes of length $n$ and distance $d$ is a set of permutations on $n$ symbols, where the distance between any two elements in the set is at least $d$. Subgroup permutation codes are permutation codes with the property that the elements are closed under the operation of composition. In this paper, under the distance metric $\ell_{\infty}$-norm, we prove that finding the minimum weight codeword for subgroup permutation code is NP-complete. Moreover, we show that it is NP-hard to approximate the minimum weight within the factor $\frac{7}{6}-\epsilon$ for any $\epsilon>0$.
\end{abstract}

\section{Introduction}

Permutation codes of length $n$ are subsets of  all permutations over $\{1,\dots,n\}$. We say a permutation code $C$ has minimum distance $d$ under some metric $\delta(\cdot,\cdot)$ if for any pair of distinct permutations $\pi$ and $\rho$ in $C$, $\delta(\pi,\rho)\ge d$. Recently, permutation codes have been found to be useful in several applications in various areas such as power line communication (see \cite{Shum02}, \cite{VH00}, \cite{VHW00}, and \cite{Vinck00}), multi-level flash memories (e.g. \cite{Jiang1}, \cite{Jiang2}, and \cite{TS09}), and cryptography (see \cite{ST09}). For these applications, researchers mainly focus on creating permutation codes within certain distance $d$ under Hamming distance, Kendall's tau distance, Chebyshev distance and other metrics which are meaningful for particular applications. 

We use $S_n$ to represent all of the permutations over $\{1,\dots,n\}$. $S_n$ is also called the symmetric group in Algebra. In this paper, we focus on permutation codes, which  also form a subgroup of $S_n$. We call them subgroup codes. A subgroup code $C$ is often defined by a generator set $\{\pi_1,\dots,\pi_k\}$ and all permutations in $C$ can be written in a sequence of compositions of elements in the generator set. This is  similar to linear codes which are subspaces of $\mathbb{F}^n$ for some finite field $\mathbb{F}$ and positive integer $n$, and lattices which are subgroups of $\mathbb{R}^n$ under the vector addition for some positive integer $n$. 

It is natural to ask how to determine the minimum distance of a code and to compute the closest codeword for a certain received string. Both problems have analogous versions for linear codes and lattices. For a linear code, it is to determine the  minimum distance while given the generator matrix of the code. This problem under Hamming distance has been proved to be NP-complete by Vardy \cite{Var97}. The analogous problem of the latter for linear codes under Hamming distance is also NP-hard by Arora et al \cite{ABSS97}. The analogous problems for lattices are the shortest lattice vector problem (SVP) and closest vector problem (CVP). SVP under $\ell_p$-norm is NP-hard, even for approximating within $p^{1-\epsilon}$ for any $\epsilon>0$ \cite{Khot05}. SVP under Chebyshev distance is also NP-hard, even for approximating within $n^{1/\log\log n}$ factor for any $\epsilon>0$ \cite{Din03}. For the subgroup permutation code version, both problems are proved to be NP-complete under many metrics, such as Hamming distance, $\ell_p$-norm, Kendall's tau, etc \cite{CW10,BCW09}. 
However, for Chebyshev distance ($\ell_{\infty}$-norm), the NP-completeness proof by Cameron and Wu \cite{CW10} fell apart
on some instances.

For right-invariant metrics, the minimum distance of subgroup permutation codes is equivalent to finding the minimum weight permutation $\pi$, where the weight of $\pi$ is defined as the distance between $\pi$ and the identity. In this paper, we focus on the complexity of the minimum weight problem for the subgroup permutation codes. We give a correct reduction to prove the NP-hardness of this problem. Moreover, we show that it is NP-hard to approximate within $\frac{7}{6}-\epsilon$ for any $\epsilon>0$. 
Our result suggests that there does not exist an efficient method which can decide  the minimum distance of an arbitrary subgroup permutation code. For example, in Tamo and Schwartz's work \cite{TS09}, they constructed some subgroup permutation codes having a minimum distance larger than they proved, but they could not give the  minimum distance explicitly with an efficient method. However, there are still some permutation codes coming with predetermined minimum distance, efficient encoding and decoding algorithms, such as in \cite{LTT08}, \cite{KLTT10} and \cite{ST09}. The situation of subgroup permutation codes is just similar to linear codes.
The rest of the paper is organized as follows.  We define some notations in Section \ref{prelim}.   The reduction is given in
Section \ref{reduc}.  Finally Section \ref{concl} concludes the paper.

\section{Preliminary}\label{prelim}

We use $[n]$ to indicate the set $\{1,\ldots,n\}$. A permutation $\pi$ over $[n]$ is a bijective function from $[n]$ to $[n]$. 
There are several representations for a permutation. In this paper, we use a truth table to denote a permutation $\pi=[\pi(1),\dots,\pi(n)]$, which can be written as the product of cycles. A cycle $(p_0,\dots,p_{k-1})$ represents a permutation 
putting the $p_i$-th entry of the input to the $p_{i+1}$-th entry of the output for $i\in\mathbb{Z}_k$. Any permutation can be written in the form of product of disjoint cycles. For example, $\pi=[2,3,1,4,6,5]=(1,3,2)(4)(5,6)$. Usually, we  ignore the cycles with only one element, therefore $[2,3,1,4,6,5]=(1,3,2)(5,6)$. 

Let $S_{n}$ denote the set of all permutations over $[n]$. It is well known that $S_n$ is a group with the composition operation. We define the product of permutations $f$ and $g\in S_n$ as $fg=[f(g(1)),\dots,f(g(n))]$. The identity permutation in $S_{n}$ is $e=[1,\dots,n]$. We say that $\{\pi_1,\dots,\pi_k\}$ is a generator set for a subgroup $H\subseteq S_n$, if every permutation $\pi\in H$ can be written as a product of a sequence of compositions from elements in the generator set.
For two permutations $\pi$ and $\rho$ over $[n]$, their Chebyshev distance is defined as $\ell_\infty(\pi,\rho)=\max_{i\in[n]}|\pi(i) -\rho(i)|$. Note $\ell_\infty$ is a right-invariant metric, i.e., for permutations $\pi$, $\rho$, and $\tau$, we have $\ell_\infty(\pi,\rho)=\ell_\infty(\pi\tau,\rho\tau)$. 

We say that a permutation $\pi$ has weight $w$ under right-invariant metric $\delta$ if $\delta(e,\pi)=w$. Now we define the minimum weight  problem of subgroup permutation  code under  Chebyshev metric, and we call it MINWSPA for short. 
\begin{definition}(MINWSPA)
Given a generator set $\{g_1,\dots,g_k\}$ for a subgroup $H$ of $S_n$ and an integer $B$, determine if there exists a permutation $\pi\in H$ that has a non-zero weight $w\le B$.
\end{definition}

Klein four-group is the building block of our proofs. It is defined as $K_4=\{e,\kappa_1,\kappa_2,\kappa_3\}$, where $\kappa_1=(1,2)(3,4)$, $\kappa_2=(1,3)(2,4)$, and $\kappa_3=(1,4)(2,3)$. 
\begin{table}[htdp]
\caption{Operation of Klein four-group.}
\begin{center}
$\begin{array}{|c||c|c|c|c|}\hline
\circ & e &\kappa_1&\kappa_2&\kappa_3\\\hline\hline
e & e&\kappa_1&\kappa_2&\kappa_3\\\hline
\kappa_1&\kappa_1&e&\kappa_3&\kappa_2\\\hline
\kappa_2&\kappa_2&\kappa_3&e&\kappa_1\\\hline
\kappa_3&\kappa_3&\kappa_2&\kappa_1&e\\\hline
\end{array}$
\end{center}
\label{group}
\end{table}
Its operation is shown in Table \ref{group}.
It is clear that $K_4$ is commutative and $\ell_\infty(e,\kappa_i)=i$ for $i\in\{1,2,3\}$.
We also use {\em shift} and {\em stretch} operations for constructing permutations. They may involve some elements of large indices. We assume that these operations are only applied on permutations over a sufficiently large symbol set. Shifting a cycle $(p_1,\dots,p_k)$ is to add the same number
to each entry of it. For example, if we shift $(1,2,3)$ with $5$, then we get $(6,7,8)$. We denote the shift operation as $s_r(\pi)$, which shifts all cycles in $\pi$ with the number $r$. For example $s_4(\kappa_1)=(1+4,2+4)(3+4,4+4)=(5,6)(7,8)$. This operation does not change the weight since the distance is preserved. 

Stretching a cycle $(p_1,\dots,p_k)$ is to multiply each entry by the same number. For example, if we stretch $(1,2,3)$ by $2$ then we have $(2,4,6)$.  We denote the stretch operation as $a_t(\pi)$ which stretches all cycles in $\pi$ by the number $t$ and then shifts the cycles such that the smallest  symbol is down to $1$. 
The distance is amplified $t$ times, and so is the weight of the cycles.
 For example $a_2(\kappa_1)=s_{-1}((1\cdot2,2\cdot2)(3\cdot2,4\cdot2))=(1,3)(5,7)$, and similarly $a_2(\kappa_2)=(1,5)(3,7), a_2(\kappa_3)=(1,7)(3,5)$. This operation amplifies the weight 2 times.
Observe that if $\{i, j, k\}=\{1, 2, 3\}$, then  $s_r(a_t(\kappa_i))s_r(a_t(\kappa_j))=s_r(a_t(\kappa_k))$ and $s_r(a_t(\kappa_i))s_r(a_t(\kappa_i))=s_r(a_t(e))$.  I.e, the shift and stretch operations  preserve the property of Klein four-group.

\section{Reduction}\label{reduc}

In this section we give a reduction from Not-All-Equal-SAT (NAESAT) to MINWSPA. Cameron and Wu\cite{CW10} gave a proof by a reduction from NAESAT to MINWSPA, but their construction fell apart on  $\ell_\infty$-norm for some instance, which is shown in Appendix A.
We give the formal definition of Not-All-Equal-SAT problem as follows.
\begin{definition}(NAESAT)
Given a boolean formula $\phi$ in conjunctive normal form, which consists of $m$ exact-3-literal clauses $c_1,\dots,c_m$ of over $n$ variables $x_1,\dots,x_n$, decide whether there exists an assignment $\sigma$ such that for every clause $c$, not all literals in $c$ are assigned to the same truth value.
\end{definition}

To construct the corresponding generator set from an NAESAT instance $\phi$
, we define three kinds of permutation gadgets for the clauses, variables and the truth assignment over $[48m+18n]$. Our goal is mapping truth assignments  for $\phi$ to permutation codewords in the corresponding subgroup permutation code. Moreover, the codewords converted from satisfying assignments have less weight than the other codewords, except
the identity. Hence, we can determine whether $\phi$ is satisfiable from the minimum weight of the corresponding subgroup permutation code.

The clause gadgets permute $1,\dots,48m$, which are derived from the work by Cameron and Wu\cite{CW10}.
The main idea of the clause gadget is to assure that all literals are not assigned to the same value. 
For convenience, we also express the following permutations with the shift and stretch operations. Let

\begin{tabbing}
$h_1$\=$=$\=$(1,3)(5,7)(2,4)(6,8)(9,13)(11,15)(10,14)(12,16)$\\
\>\>$(17,23)(19,21)(18,24)(20,22)$\\
\>$=$\>$a_2(\kappa_1)s_{1}(a_2(\kappa_1))s_{8}(a_2(\kappa_2))s_9(a_2(\kappa_2))$\\
\>\>$s_{16}(a_2(\kappa_3))s_{17}(a_2(\kappa_3))$,\\
$h_2$\>$=$\>$(1,5)(3,7)(2,6)(4,8)(9,15)(11,13)(10,16)(12,14)$
\\\>\>$(17,19)(21,23)(18,20)(22,24)$\\
\>$=$\>$a_2(\kappa_2)s_{1}(a_2(\kappa_2))s_{8}(a_2(\kappa_3))s_9(a_2(\kappa_3))$\\
\>\>$s_{16}(a_2(\kappa_1))s_{17}(a_2(\kappa_1))$,\\
$h_3$\>$=$\>$(1,7)(3,5)(2,8)(4,6)(9,11)(13,15)(10,12)(14,16)$\\
\>\>$(17,21)(19,23)(18,22)(20,24)$\\
\>$=$\>$a_2(\kappa_3)s_{1}(a_2(\kappa_3))s_{8}(a_2(\kappa_1))s_9(a_2(\kappa_1))$\\
\>\>$s_{16}(a_2(\kappa_2))s_{17}(a_2(\kappa_2))$,\\
$g$\>\=$=$\=$(1,8)(2,7)(3,6)(4,5)(9,16)(10,15)(11,14)(12,13)$\\
\>\>\>$(17,24)(18,23)(19,22)(20,21)$\\
\>$=$$\left(\prod_{i=1}^{12}(2i-1,2i)\right)a_2(\kappa_3)s_{1}(a_2(\kappa_3))$\\
\>\>\>$s_{8}(a_2(\kappa_3))s_9(a_2(\kappa_3))s_{16}(a_2(\kappa_3))s_{17}(a_2(\kappa_3))$.
\end{tabbing}

Note that, for each of the above permutations, the first 4 pairs permute 1-8, 
the next 4 pairs permute 9-16, and the
last 4 pairs permute 17-24. The operations among these 4 permutations are commutative. 
It is clear that $h_1$, $h_2$ and $h_3$ each has weight 6, and the weight of $g$ is $7$.
The weights of $gh_1$, $gh_2$ and $gh_3$ are all 5.

The clause gadgets corresponding to the $k$-th literal of the $j$-th clause assigned true and false are defined as $h_{j,k,T}=s_{48(j-1)}(gh_k)$ and $h_{j,k,F}=s_{48(j-1)+24}(gh_k)$, respectively. For every $j\in[m]$ and $t\in\{T,F\}$, we have:
\begin{itemize}
\item For every $k\in[3]$, $h_{j,k,t}$ has weight $5$.
\item For distinct $k,k'\in[3]$, $h_{j,k,t}h_{j,k',t}$ has weight $6$.
\item $h_{j,1,t}h_{j,2,t}h_{j,3,t}$ has weight $7$.
\end{itemize}

The second kind is the variable gadget which assures that no variable is assigned both true and false. They permute 
elements $48m+1,\dots,48m+10n$. Let
\[v_T=(1,4)(7,10)=a_3(\kappa_1), v_F=(1,7)(4,10)=a_3(\kappa_2).\]
Note that $v_Tv_F=a_3(\kappa_3)$.
The weights of $v_T$, $v_F$, and $v_Tv_F$ are 3, 6, and 9, respectively.
The variable gadgets corresponding to  $x_i$ 
are defined as 
$v_{i,T}=s_{b_1+10i}(v_T)$ and $v_{i,F}=s_{b_1+10i}(v_F)$ where $b_1=48m-10$.

The third kind is the assignment gadget which assures if $x_i$ is assigned, then $x_{i+1}$ and $x_{i-1}$ are assigned, too, where $i\pm 1 \in \mathbb{Z}_n$ and $x_0\equiv x_n$. They permute elements $48m+10n+1,\dots,48m+18n$.
We use the following permutation to give a chain reaction, i.e., if there is any missing gadget, then the distance will
deviate significantly. The assignment gadget for $x_i$ is defined as $u_i=s_{b_2+8(i-1)}((1,8))s_{b_2+8i}((1,8))$ for $i<n$ and $u_n=s_{b_2+8n}((1,8))s_{b_2}((1,8))$ where $b_2=48m+10n$. For convenience, we also use $u_0$ as the alias of the gadget $u_n$.

Now we give the polynomial-time mapping function from NAESAT to MINWSPA. Let \[\begin{array}{c}P=\{(i,j,k):x_i\mbox{ is the }k\mbox{-th literal in }c_j\},\\
Q=\{(i,j,k):\bar{x}_i\mbox{ is the }k\mbox{-th literal in }c_j\}.\end{array}\] For the $i$-th variable $x_i$, we define
\[g_i=v_{i,T}u_i\left(\prod_{(i,j,k)\in P}h_{j,k,T}\right)\left(\prod_{(i,j,k)\in Q}h_{j,k,F}\right),\]
\[g'_i=v_{i,F}u_i\left(\prod_{(i,j,k)\in P}h_{j,k,F}\right)\left(\prod_{(i,j,k)\in Q}h_{j,k,T}\right).\]
The generator set is $\{g_i,g'_i:i\in[n]\}$. The scheme above can be done in polynomial time, since $|P|+|Q|=3m$ and the size of each gadget is at most $O((48m+18n)\log(48m+18n))$.

Let $H$ be the subgroup generated by $\{g_i, g'_i:i\in[n]\}$, i.e., $H=\langle g_i, g'_i: i\in[n]\rangle$. We can obtain 
a permutation \[\pi=\left(\prod_{\sigma(x_i)=T}g_i\right)\left(\prod_{\sigma(x_i)=F}g'_i\right)\in H\] from an assignment $\sigma$ for $\phi$. By the following two lemmas, we show that there must exist a non-identity permutation of minimum weight which is constructed from an assignment.

\begin{lemma}\label{assignment}
The permutations mapped from satisfying assignments have weight 6 and
the permutations mapped from unsatisfying assignments have weight $7$. 
\end{lemma}
\begin{proof}
Let $\pi$ be the permutation obtained from an assignment $\sigma$ of $\phi$.
Note that the elements permuted by the clause gadgets, variable gadgets, and assignment gadgets are disjoint.
Therefore we can discuss the weight of them separately in three categories. First, we look at the elements permuted by variable gadgets $v_{i,T}$ and $v_{i,F}$ for $i\in[n]$. Only $g_i$ and $g'_i$ can alter these elements, and $\pi$ has exactly one of them. Thus, the difference between $\pi$ and $e$ on these elements is at most $6$. Next we turn to assignment gadget $u_i$ for some $i\in[n]$. Without loss of generality, we assume $u_i=(n_1,n_1+7)(n_2,n_2+7)$, $u_{i-1}$ contains $(n_1,n_1+7)$ and $u_{i+1}$ contains $(n_2,n_2+7)$. Since $\pi$ has exactly one of $g_i$ and $g'_i$ for every $i\in[n]$, both $(n_1,n_1+7)$ and $(n_2,n_2+7)$ appear twice in the construction of $\pi$. Moreover, they are the only cycles covers $n_1$, $n_1+7$, $n_2$, and $n_2+7$. Thus, $\pi$ does not affect these elements and there is no difference between $\pi$ and $e$ on them.

At last, we observe clause gadgets $h_{j,k,T}$ and $h_{j,k,F}$. We claim that for $j\in[m]$ and $k\in[3]$, exactly one of $h_{j,k,T}$ and $h_{j,k,F}$ appears in $\pi$. Assume $x_i$ is the $k$-th literal of the $j$-th clause. If $\sigma(x_i)=T$, then $h_{j,k,T}$ is picked by the definition of $g_i$, otherwise $\pi$ picks $h_{j,k,F}$. It is similar for the case that $\bar{x}_i$ is the $k$-th literal in the $j$-th clause. Thus, for every $j\in[m]$, $\pi$ picks exactly three out of $h_{j,1,T}$, $h_{j,2,T}$, $h_{j,3,T}$, $h_{j,1,F}$, $h_{j,2,F}$, and $h_{j,3,F}$. Let $A_j=\{h_{j,k,T}:k\in[3]\}$ and $B_j=\{h_{j,k,F}:k\in[3]\}$. In the following, we discuss how these gadgets affect the distance between $\pi$ and $e$.
\begin{enumerate}
\item If the gadgets in $A_j$ are not picked at all, then the elements permuted by $A_j$ remain the same as $e$. But this implies all gadget in $B_j$ are picked, then the elements permuted by $B_j$ are permuted with a shift of $g$. The distance is $\max\{0,7\}=7$.
\item If one of $A_j$ and two of $B_j$ are picked, then the elements permuted by $A_j$ and $B_j$ are in the form of $gh_{k}$ and $h_{k'}$ for some $k,k'\in[3]$, respectively. The distance is $\max\{5,6\}=6$. 
\item If two of $A_j$ and one of $B_j$ are picked, then, similar to 2), the distance is 6.
\item If all of $A_j$ are picked, then,  similar to 1), the distance is $7$.
\end{enumerate}
Note that the first or last cases above happen if and only if $\sigma$ is not a satisfying assignment. Since distances of the clause gadgets dominate the distance over the other gadgets, we conclude that $\pi$ has weight 6 if $\sigma$ is satisfying;  $7$ otherwise. 
\end{proof}

\begin{lemma}\label{nonassignment}
The other non-identity permutations in $H$ have weight at least $7$. 
\end{lemma}

\begin{proof}
Since all gadgets are commutative, we can express any permutation $\pi\in H$ into a product of powers of generators, i.e., $\pi=g_1^{z_1}(g'_1)^{z'_1}\cdots g_n^{z_n}(g'_n)^{z'_n}$. Since every gadget is the inverse of itself, we assume $z_1,z'_1,\dots,z_n,z'_n\in\{0,1\}$ without loss of generality. A permutation converted from an assignment must choose either $g_i$ or $g'_i$, for every $i\in[n]$, i.e., $z_i+z'_i=1$ for $i\in[n]$. So we discuss the following two cases.
\begin{itemize}
\item If there exists some $i$ such that $z_i+z'_i=2$, then $\pi$ picks both $g_i$ and $g'_i$. In this case, the elements corresponding to $v_{i,T}$ and $v_{i,F}$ are permuted into the form of $v_Tv_F$, which has weight $9$.
\item For every $i\in[n]$, $z_i+z'_i\neq 2$. Because $\pi\neq e$ and $\pi$ is not converted from an assignment, there are $i_0$ and $i_1$ such that $z_{i_0}+z'_{i_0}=0$ and $z_{i_1}+z'_{i_1}=1$. 
Now recall that for $i<n$, $u_i=s_{b_2+8(i-1)}((1,8))s_{b_2+8i}((1,8))$ and $u_n=s_{b_2+8n}((1,8))s_{b_2}((1,8))$. Without loss of generality, we can assume that $i_0=i_1-1$. As a consequence, $u_{i_0}$ and $u_{i_1}$ are the only two gadgets permuting $b_2+8i_1-7$ and $b_2+8i_1+7$. Since $z_{i_0}+z'_{i_0}+z_{i_1}+z'_{i_1}=0+1=1$, $\pi$ picks exactly one of $g_{i_0}$, $g'_{i_0}$, $g_{i_1}$ and $g'_{i_1}$. $b_2+8i_1-7$ and $b_2+8i_1$ must be swapped by $\pi$, hence $\pi$ has distance at least $7$ in this case.
\end{itemize}
The non-identity permutations, which are not in the two cases, have exactly one of $g_i$ and $g'_i$ for every $i\in[n]$, and these can be obtained from assignments. We conclude the lemma is true.
\end{proof}

With the two lemmas above, we prove the following theorem.

\begin{theorem}\label{main}
Let $H$ be the group generated by $\{g_1,g'_1\dots,g_n,g'_n\}$ which is mapped from a NAESAT instance $\phi$. If $\phi$ is satisfiable then $H$ has minimum weight $6$, otherwise $H$ has minimum weight $7$. 
\end{theorem}

\begin{proof}
For satisfiable $\phi$, there exists a satisfying assignment $\sigma$. We can convert $\sigma$ into $\pi$, and $\pi$ has weight 6 by lemma \ref{assignment}. For unsatisfiable $\phi$, all assignments are not satisfying. So every permutation converted from an assignment has weight $7$, and the other non-identity permutations have weight at least $7$. Thus, we conclude $H$ has minimum weight $7$. 
\end{proof}
From the above theorem we have an immediate inapproximable result.
We say that an algorithm $A$ is an $r$-approximate algorithm for a minimization problem if $A$ always outputs a feasible solution whose cost is no more than $r$ times of the minimum cost on any input. Note that $A$ cannot output an answer whose cost is less than the minimum cost, since it is not a feasible solution. Since NAESAT is an NP-complete problem, we have the following corollary as an immediate result of theorem \ref{main}.

\begin{corollary}
MINWSPA is NP-complete. Moreover, it is NP-hard to approximate within $\frac{7}{6}-\epsilon$ for any $\epsilon>0$.
\end{corollary}

\begin{proof}
MINWSPA is in NP, since we can finish computing the weight of any permutation $\pi$ and verifying if $\pi$ is in the subgroup by Schreier-Sims algorithm \cite{Sims70} in polynomial time.
By theorem \ref{main}, $\phi$ is satisfiable if and only if the corresponding subgroup $H$ has minimum weight at most 6. Hence, we can conclude MINWSPA is NP-complete. Now, assume we have a polynomial time $(\frac{7}{6}-\epsilon)$-approximate algorithm $A$. We can construct a polynomial time algorithm to solve NAESAT.
\begin{enumerate} 
\item Construct the subgroup $H$ from $\phi$ and run $A(H)$.
\item If $A(H)$ outputs a number no more than $7-6\epsilon$, then accept $\phi$, otherwise reject.
\end{enumerate}
Any satisfiable $\phi$ will be accepted, and all unsatisfiable $\phi$'s will be rejected, since an approximate algorithm cannot give an answer less than the minimum solution which is $7$.
\end{proof}

\section{Conclusion}\label{concl}
We show that MINWSPA is NP-complete. 
It implies that the minimum weight problem of permutation codes under the well known metrics are all NP-complete.
For the case of $\ell_\infty$-metric, we also prove that there is no $\left(\frac{7}{6}-\epsilon\right)$-approximate algorithm for any $\epsilon>0$ unless P$=$NP. We believe that the minimum weight problems under other metrics also have inapproximable results, however, they still remain open. 
Our inapproximable result still has room for improvement.  It is interesting to find better approximation algorithm with some constant $c>\frac{7}{6}$.

\begin{appendix}\label{app}
\subsection{Cameron-Wu's reduction}
The reduction in Cameron and Wu's work \cite{CW10} uses only two kinds of gadgets. The variable gadget $v_i$ for the $i$-th variable is $(2i-1,2i)$. The clause gadget $h_{j,k}$ for the $k$-th literal in the $j$-th clause is defined as $s_{2n+24(j-1)}(h_k)$ where $h_1$, $h_2$, $h_3$ and $g$ are the same as in this paper. 
The generators are defined as
\[g_i=v_i\left(\prod_{(i,j,k)\in P}s_{2n+24(j-1)}(h_k)\right),\]\[g'_i=v_i\left(\prod_{(i,j,k)\in Q}s_{2n+24(j-1)}(h_k)\right),\]
where $P,Q$ are the same sets as in our reduction. They also construct a generator $g_c=\prod_{j\in[m]}s_{2n+24(j-1)}(g)$ acting as $g$ on every clause gadget. Their construction does not work in the following instance. Let $\phi=(x_1\vee x_2\vee x_2)\wedge(\bar{x}_1\vee x_2\vee x_2)$. By their construction, subgroup $G$ is generated by
\begin{tabbing}
$g_1=v_1h_{1,1}=(1,2)s_4(h_1)$,\\
$g'_1=v_1h_{2,1}=(1,2)s_{28}(h_1)$,\\
$g_2$ \=$=v_2h_{1,2}h_{1,3}h_{2,2}h_{2,3}$\\
\>$=(3,4)s_4(h_2)s_4(h_3)s_{28}(h_2)s_{28}(h_3)$\\
\>$=(3,4)s_4(h_1)s_{28}(h_1)$,\\
$g'_2=v_2=(3,4),$\\
$g_c=s_{4}(g)s_{28}(g).$
\end{tabbing}
Note that $\phi$ is an unsatisfiable formula for NAESAT.  According to their proof of Theorem 18\cite{CW10}, elements of $G$ should not have weight 5, since $\phi$ is unsatisfiable. But $g_cg_1g'_1=s_{4}(gh_1)s_{28}(gh_1)$ has weight 5. 
Therefore, we need to design the gadgets more carefully to prove that MINWSPA is NP-complete.
\end{appendix}


\begin{thebibliography}{1}

\bibitem{ABSS97}
S. Arora, L. Babai, J. Stern, Z. Sweedyk,
``The Hardness of Approximate Optima in Lattices, Codes, and Systems of Linear Equations,''
\emph{Journal of Computer and System Science}, vol. 54, 1997, pp. 317-331.

\bibitem{BMT78}
E. R. Berlekamp, R. J. McEliece, H. C.A. van Tilborg, 
``On the inherent intractibility of certain coding problems,'' 
IEEE Trans. Inform. Theory IT-24, 1978, pp. 384-–386.

\bibitem{BCW09}
C. Buchheim, P. J. Cameron, T. Wu,
``On the subgroup distance problem,''
\emph{Discrete Mathematics}, vol. 309, pp. 962--968, 2009.

\bibitem{CW10} 
P. J. Cameron, T. Wu,
``The complexity of the weight problem for permutation and matrixgroups,''
\emph{Discrete Mathematics}, vol. 310, pp. 408--416, 2010.








\bibitem{Din03}
I. Dinur, ``Approximating SVP$_\infty$ to within almost polynomial factors is NP-hard,''
\emph{Combinatorica}, vol. 23, 2003, pp. 205-–243.





\bibitem{Jiang1} 
A. Jiang, R. Mateescu, M. Schwartz, J. Bruck, 
``Rank Modulation for Flash Memories,''
in \emph{Proc. IEEE Internat. Symp. on Inform. Th.}, 2008, pp. 1731-1735.

\bibitem{Jiang2} 
A. Jiang, M. Schwartz, J. Bruck, 
``Error-Correcting Codes for Rank Modulation,''
in \emph{Proc. IEEE Internat. Symp. on Inform. Th.}, 2008, pp. 1736-1740.



\bibitem{Khot05}
S. Khot, ``Hardness of Approximating the Shortest Vector Problem in Lattices,''
\emph{J. ACM}, Vol. 52, No. 5, 2005, pp. 789--808.

\bibitem{KLTT10}
T. Klove, T.-T. Lin, S.-C. Tsai, W.-G. Tzeng, 
``Permutation arrays under the Chebyshev distance,''
in \emph{Proc. IEEE Trans. on Inform. Th.}, accepted and to appear, 2010.







\bibitem{LTT08} 
T.-T. Lin, S.-C. Tsai, W.-G. Tzeng, 
``Efficient Encoding and Decoding with Permutation Arrays,''
in \emph{Proc. IEEE Internat. Symp. on Inform. Th.}, 2008, pp. 211-214.

\bibitem{ST09}
M.-Z. Shieh, S.-C. Tsai,
``Decoding Frequency Permutation Arrays under Infinite norm,''
in \emph{Proc. IEEE Internat. Symp. on Inform. Th.}, 2009, pp. 2713--2717.

\bibitem{Shum02} 
K. W. Shum, 
"Permutation coding and MFSK modulation for frequency selective channel,''
\emph{IEEE Personal, Indoor and Mobile Radio Communications}, vol. 13, pp. 2063--2066, Sept. 2002.

\bibitem{Sims70}
C. Sims, ``Computational methods in the study of permutation groups'', \emph{Computational Problems in Abstract Algebra}, pp. 169-183, Pergamon, Oxford, 1970.






\bibitem{TS09}
I. Tamo, M. Schwartz,
``Correcting Limited-Magnitude Errors in the Rank-Modulation Scheme,''
arXiv:0907.3387v2. 




\bibitem{Var97}
A. Vardy, 
``The intractability of computing the minimum distance of a code,'' 
IEEE Trans. Inform. Theory 43, 1997, pp. 1757-1766.

\bibitem{Vinck00}
A. J. H. Vinck, 
``Coded modulation for powerline communications,''
\emph{Proc. Int. J. Electron. Commun}, vol. 54, pp. 45-49, 2000.

\bibitem{VH00}
A. J. H. Vinck, J. H\"{a}ring, 
``Coding and modulation for power-line communications,''
in \emph{Proc. Internat. Symp. on Power Line Commun.}, Limerick, Ireland, April 2000.

\bibitem{VHW00}
A. J. H. Vinck, J. H\"{a}ring, T. Wadayama, 
``Coded M-FSK for power line communications,''
in \emph{Proc. IEEE Internat. Symp. on Inform. Th.}, 2000, p.137.




\end{thebibliography}
\end{document}